\documentclass[prb,preprint,amsmath,amssymb,aps,floatfix,longbibliography]{revtex4-1}

\usepackage[utf8]{inputenc}
\usepackage[T1]{fontenc}
\usepackage[english]{babel}
\usepackage{amsthm}
\usepackage{amscd}
\usepackage{amsfonts}         
\usepackage{amsmath}
\usepackage{amssymb}
\usepackage[hidelinks]{hyperref}
\usepackage{enumerate}
\usepackage{tikz}
\usepackage{tikz-cd}
\usepackage{mathrsfs}  
\usepackage{todonotes}
\usepackage[normalem]{ulem}
\usepackage[pagewise, displaymath, mathlines]{lineno}

\newcommand{\Rb}{\mathbb{R}}

\newcommand{\Zb}{\mathbb{Z}}

\newcommand{\SO}{\mathrm{SO}}

\newcommand{\Spin}{\mathrm{Spin}}
\newcommand{\diag}{\mathrm{diag}}
\newcommand{\red}{\mathrm{red}}
\newcommand{\blue}{\mathrm{blue}}
\newcommand{\green}{\mathrm{green}}

\newtheorem{theo}{Tplottin ubuntuheorem}
\theoremstyle{plain}
\newtheorem{thm}[theo]{Theorem}
\newtheorem{lem}[theo]{Lemma}
\newtheorem{prop}[theo]{Proposition}

\newtheorem*{thm*}{Theorem}
\newtheorem*{lem*}{Lemma}
\newtheorem*{prop*}{Proposition}
\newtheorem*{cor*}{Corollary}

\theoremstyle{definition}
\newtheorem{defn}[theo]{Definition}

\date{\today}

\begin{document}
\title{Topologically protected vortex knots and links}
\author{Toni Annala}\email{toni.annala@aalto.fi}
\affiliation{QCD Labs, QTF Centre of Excellence, Department of Applied Physics, Aalto University, P.O. Box 13500, FI-00076 Aalto, Finland}
\affiliation{Department of Mathematics, University of British Columbia, 1984 Mathematics Rd, Vancouver, BC V6T 1Z2, Canada}
\author{Roberto Zamora-Zamora}
\affiliation{QCD Labs, QTF Centre of Excellence, Department of Applied Physics, Aalto University, P.O. Box 13500, FI-00076 Aalto, Finland}
\author{Mikko M{\"o}tt{\"o}nen}
\affiliation{QCD Labs, QTF Centre of Excellence, Department of Applied Physics, Aalto University, P.O. Box 13500, FI-00076 Aalto, Finland}
\date{\today}

\begin{abstract}\textbf{
In 1869, Lord Kelvin found that the way vortices are knotted and linked in an ideal fluid remains unchanged in evolution, and consequently hypothesized atoms to be knotted vortices in a ubiquitous ether, different knotting types corresponding to different types of atoms\cite{thomson:1869}. Even though Kelvin's atomic theory turned out incorrect, it inspired several seminal results, such as the mathematical theory of knots \cite{atiyah:1990,rolfsen:2003} and the investigation of knotted structures that naturally arise in physics\cite{faddeev:1997}. However, in previous studies, knotted and linked structures have been found to untie via local cut-and-paste events referred to as \emph{reconnections}\cite{cirtain:2013,kleckner:2016}. Here, in contrast, we construct knots and links of non-Abelian vortices that are \emph{topologically protected} in the sense that they cannot be dissolved employing local reconnections and strand crossings. We prove such topological stability by introducing an invariant of links colored by the elements of the quaternion group $Q_8$. Importantly, the topologically protected links are supported by a variety of physical systems such as dilute Bose--Einstein condensates. Thus, our results predict the existence of topologically stable knotted structures in condensed-matter systems. To the best of our knowledge, this kind of stability has not been found previously.  We expect our results to lead to discoveries of analogous linked and knotted structures in other physical systems including, in addition to ultracold atomic gases, at least liquid crystals and cosmology. Interestingly, we also propose a classification scheme for topological vortex links, in which two structures are considered equivalent if they  differ from each other by a sequence of topologically allowed reconnections and strand crossings, in addition to the typical continuous transformations. This leads to a simpler classification than the schemes in which the core topology is constant\cite{machon:2014}: for example, there are essentially only three types of nontrivial $Q_8$-colored links. In the future, the detailed study of these novel structures calls for new mathematical tools, which we expect to lead to advances in knot theory and in low-dimensional topology.
}\end{abstract}

\maketitle

Mathematically, knots and links are closed loops and configurations thereof  in a three-dimensional space, respectively\cite{rolfsen:2003}. There are numerous physical systems that support knotted structures. Examples include the disclination lines of liquid crystals \cite{smalyukh:2009,tkalec:2011,sec:2014}, the cores of vortices in water\cite{kleckner:2013} and in superfluids\cite{kleckner:2016}, strands of DNA\cite{han:2010}, Skyrmion cores in classical field theory \cite{faddeev:1997}, and the quantum knot in the polar phase of a spin-1 Bose--Einstein condensate \cite{hall:2016}. Additionally, there are deep connections between the mathematical theory of knots and statistical mechanics, topological quantum computing, and quantum field theories\cite{jones:1987,witten:1989,kauffman:1994}. In mathematics, knots play a significant role, for example, in the surgery theory of three-dimensional manifolds\cite{prasolov:1996}.


A non-trivially knotted loop tied from a physical string, or more generally a link tied from several loops, cannot be untied without cutting at least one string; the knot and link structures are robust. Similarly, knots and links tied from vortex lines in a frictionless ideal fluid remain forever knotted\cite{thomson:1869}, leading to a conserved quantity, \emph{helicity}\cite{moffatt:1969}, which measures the total knottedness and linkedness of a vortex configuration. However, it has been observed that even a small amount of dissipation is enough to cause spontaneous untying of knotted vortices owing to local reconnection events \cite{cirtain:2013, kleckner:2016}. In addition, another local modification, namely strand crossing\cite{poenaru:1977}, may promote the untying process of vortex knots. These observations naturally lead to the following fundamental question: is it possible to tie a vortex knot, or a link, for which decay through local reconnection events and strand crossings is prohibited by the fundamental properties of the knot-supporting substance?


Here, we propose a class of knotted structures, tied from non-Abelian \emph{topological vortices}, that have exactly this property. We refer to such a structure as being \emph{topologically protected} against untying. Topological vortices are codimension-two defects in an ordered medium where the winding of the order parameter field about the vortex core corresponds to a non-trivial element of the \emph{fundamental group} $\pi_1$. The elements of the fundamental group $\pi_1(M,m)$ correspond to oriented loops in the order parameter space\cite{mermin:1979} $M$ that begin and end at the \emph{basepoint} $m \in M$, considered up to continuous deformations that keep the endpoints fixed. Here, the group law is provided by the concatenation of loops. The basepoint $m$ is often omitted from the notation. 

Topological vortices are non-Abelian if the fundamental group $\pi_1(M)$ is a non-Abelian group, such as the quaternion group $Q_8 = \{\pm 1, \pm i, \pm j, \pm k\}$, the group law of which is governed by the multiplication rule of quaternions. Non-Abelian vortices are known to exhibit peculiar behavior, when they interact with each other. For instance, two vortices corresponding to non-commuting elements in $\pi_1(M)$ cannot strand cross, i.e., freely pass through each other. Instead, one concludes on topological grounds alone that a \emph{vortex chord}, corresponding to the commutator of the crossing vortices in $\pi_1(M)$, forms to connect them\cite{poenaru:1977}. Precisely such phenomena are responsible for the robustness of the structures discovered in this work.

We begin by identifying topological vortex configurations for certain, physically relevant, order parameter spaces, with colored link diagrams (Fig.~\ref{fig:diagex}). Then, we identify realistic rules governing the core-topology-altering evolution of such structures. In the special case of $Q_8$-colored link diagrams, which describe topological vortex configurations in certain physical systems, the colored link diagrams and the evolution rules admit a simple graphical depiction (Fig.~\ref{fig:rules}). Using invariants of $Q_8$-colored links, we identify examples of linked structures that are robust against local reconnections and strand crossings (Fig.~\ref{fig:defects}). Interestingly, we also classify all the $Q_8$-colored links up to reconnections and strand crossings, and find a topologically stable knot for a fundamental group of permutations. 

A configuration of topological vortices is formalized as follows. The spatial extent of the ordered medium is modeled by the three-dimensional euclidean space $\Rb^3$. Homotopy theoretically, it makes no difference if $\Rb^3$ is replaced by the unit ball in $\Rb^3$. The cores of the vortices, i.e., the location where the order parameter is not well defined, form a subset $L \subset \Rb^3$ consisting of loops. In other words, the collection of cores forms a link in $\Rb^3$. The order parameter field is modeled by a continuous map $\Psi: \Rb^3 \backslash L \to M$, where $M$ is the order parameter space. The induced homomorphism between the fundamental groups $\psi: \pi_1(\Rb^3 \backslash L, x_0) \to \pi_1[M, \Psi(x_0)]$ may be described by a $\pi_1(M)$-colored link diagram as illustrated in Fig.~\ref{fig:diagex}.  If the second homotopy group $\pi_2[M, \Psi(x_0)]$, i.e., the group defined analogously to the fundamental group but with spheres instead of loops, is trivial, then the group homomorphism $\psi$ retains all the information about the homotopy class\cite{ang:2018,ang:2018b} of $\Psi$. In fact, the homotopy classes of continuous maps $\Rb^3 \backslash L \to M$ are in one-to-one correspondence with group homomorphisms $\pi_1(\Rb^3 \backslash L) \to \pi_1(M)$ up to ``change of coordinates'' (up to \emph{conjugacy}): homomorphisms $\psi, \psi': \pi_1(\Rb^3 \backslash L) \to \pi_1(M)$ correspond to the same homotopy class if and only if there exists $h \in \pi_1(M)$ such that, for all $g \in \pi_1(\Rb^3 \backslash L)$, $\psi'(g) = h \psi(g) h^{-1}$. In terms of $\pi_1(M)$-colored link diagrams, this relation considers two colored diagrams equivalent if one of them can be obtained from the other by replacing all $g_i$ with $h g_i h^{-1}$. The rules for $Q_8$-colored link diagrams admit a purely graphical presentation, as illustrated in Fig.~\ref{fig:rules}a, b.

For the stability of such structures, we make the following assumption: the core topology changes only in topologically allowed local reconnections and in topologically allowed strand crossings. In particular, we do not account for the possibility of a vortex spontaneously splitting into two vortices, or a vortex ring shrinking into a point defect. We justify this by restricting our attention to topological vortices for which such events are, for example, energetically unfavorable. Furthermore, it turns out that a decay through vortex splitting requires it to take place in a region extending from one undercrossing to another, which is unlikely provided that there are no intrinsic mechanisms driving the vortex splitting. 

A vortex reconnection  is topologically allowed if it does not lead to discontinuities in the coloring. A strand crossing is topologically allowed\cite{poenaru:1977} if the two strands correspond to commuting elements in $\pi_1(M)$. It is a consequence of the Wirtinger relations\cite{rolfsen:2003} that this is equivalent to assuming that the coloring does not change in the undercrossing. Hence, strand crossing is allowed only if no discontinuities in the coloring occur after the crossing has taken place. For $Q_8$-colored link diagrams, the local reconnection and strand-crossing rules admit simple description in terms of the graphical presentation mentioned above (Fig.~\ref{fig:rules}c).

Let us consider physical systems that consist of the cyclic or biaxial nematic phases of a spin-2 Bose--Einstein condensate, or of the biaxial nematic phase of a liquid crystal. The corresponding order parameter spaces---$M_\mathrm{C}, M_\mathrm{BN}$, and $M_\mathrm{BNLQ}$, respectively---have trivial second homotopy groups; importantly, the fundamental group of $M_\mathrm{BNLQ}$ is $Q_8$, whereas in $\pi_1(M_\mathrm{C})$ and $\pi_1(M_\mathrm{BN})$, $Q_8$ is the subgroup corresponding to topological vortices with no scalar-phase winding about the core (Supplementary Information). Hence, in each case, at least some of the topological vortex configurations admit descriptions as $Q_8$-colored link diagrams. Moreover, the evolution of these structures under strand crossings and local reconnections may be analyzed using the rules outlined in Fig.~\ref{fig:rules}c.

Next, we establish the the topological stability of the knotted structures by means of invariants of $Q_8$-colored links. The \emph{linking invariant} $l$ is an element of $\Zb_2$ that is the number of times a blue strand crosses over a red strand, modulo 2. Equivalently, it is the total Gauss linking number\cite{rolfsen:2003} between the red and the blue loops, modulo 2. In order for the bicoloring flips to lead to a consistent bicoloring, each red loop is overcrossed an even number of times by a strand of either green or blue color. Hence, the number of times a red loop is crossed over by a blue strand is equivalent, modulo 2, to the number of times it is crossed over by a green strand. In fact, one argues in a similar fashion that the linking invariant $l$ is independent of the pair of colors used to compute it. The invariant is conserved in topologically allowed local reconnections and strand crossings (Fig.~\ref{fig:rules}c), since these modifications do not alter crossings that are relevant for $l$. For a diagram of disjoint and untangled loops, this invariant is clearly zero. Hence, a configuration with a non-trivial linking invariant cannot be unknotted by local reconnections and strand crossings. 

In the Supplementary Information, we define the \emph{$Q$-invariant} which is a $\Zb_4$-valued invariant that recovers the linking invariant $l$ when reduced to modulo 2. It is also conserved in strand crossings and local reconnections, and therefore can be used as a marker of topological protection. Remarkably, up to topologically allowed local reconnections and strand crossings, there are only six non-trivial $Q_8$-colored links, two for each value $Q= [1], [2]$, and $[3]$, respectively (Fig.~\ref{fig:defects}a). Moreover, there are $2^4 = 16$ defects that have $Q=[0]$, each of which is equivalent to a disjoint union of loops of a subset of the four possible colors. Examples of topologically protected and unprotected vortex configurations are illustrated in Fig.~\ref{fig:defects}a, b.

These results on the invariants imply that a system, the fundamental group of which is described by the quaternion group, cannot support topologically protected vortices composed of a single loop. At least three different colors are needed for stability. Consequently, we have studied the permutation group $S_3$ and found that a simple trefoil knot is indeed topologically protected as illustrated in Fig.~\ref{fig:defects}c. The detailed study of this exciting structure and its realizations in physical systems is left for future work.

We expect our results to inspire a wide range of theoretical and experimental investigations. In addition to observing experimentally the proposed vortex links in nematic liquid crystals or spin-2 Bose--Einstein condensates, analogous topologically protected vortex structures can be studied both theoretically and experimentally in a wide variety of condensed-matter systems, which may stimulate the development of invariants for other types of colored links. It would be especially interesting to find a tangled vortex structure that cannot be untied even in the presence of vortex splittings which may break the links considered here as is illustrated in Extended Data Figure 7. Another interesting aspect of our work is the simplicity of the classification: instead of the infinitude of different link types, there is only a small number of $Q_8$-colored links up to strand crossings and reconnections. It would be interesting to obtain similar classifications for physically relevant groups other than the quaternion group $Q_8$. In future work, we aim to use numerical simulations to investigate the dynamics of the structures proposed in this article to verify their stability properties and provide insight how they can be prepared and observed in experiments.


\bibliographystyle{naturemag}
\bibliography{references}

\begin{figure}[h!]
\includegraphics[scale=1.5]{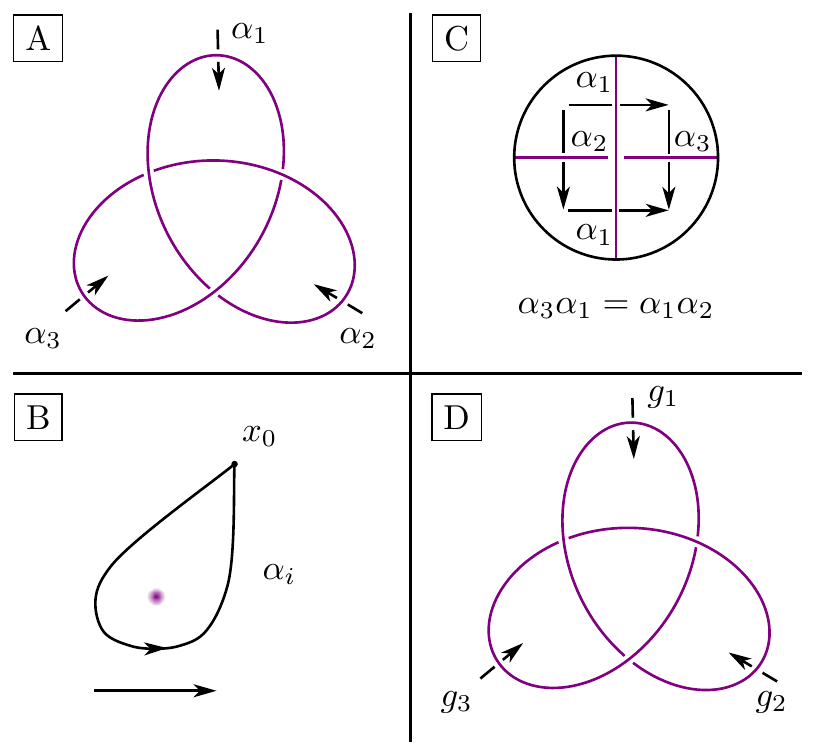}
\caption{\textbf{Wirtinger presentation of the fundamental group of the link complement and colored link diagrams.}
\textbf{a}, \textbf{b}, Link complement $\Rb^3 \backslash L$ (area outside the purple line which denotes $L$) and loops denoted by the arrows. After fixing a basepoint $x_0 \in \Rb^3 \backslash L$ above the plane of the picture \textbf{a}, each arrow describes a loop in $\Rb^3 \backslash L$ based at $x_0$, the homotopy equivalence class of which is denoted by $\alpha_i$ (\textbf{b}). Every loop that can be continuously deformed into each other without passing through $L$ and always fixed at $x_0$ are equivalent. \textbf{c}, Each crossing of the link diagram in \textbf{a} corresponds to a relation in the fundamental group $\pi_1(\Rb^3 \backslash L, x_0)=\{\alpha_i\}$, one of which is presented in \textbf{c}. The \emph{Wirtinger presentation} of $\pi_1(\Rb^3 \backslash L, x_0)$ describes it as the group, the elements of which are words on symbols $\alpha_i$ and $\alpha^{-1}_i$, modulo the relations associated to the crossings of the diagram, with the group operation corresponding to concatenation of words. 
\textbf{d}, Group homomorphism $\psi: \pi_1(\Rb^3\backslash L, x_0) \to G$ may be described by specifying the images $g_i := \psi(\alpha_i)$ as shown above. Such a picture is regarded as a \emph{($G$-)colored link diagram}. In order for the colored diagram to correspond to a well-defined homomorphism, the elements $g_i$ must satisfy the analogue of the Wirtinger relations, e.g., $g_3 g_1 = g_1 g_2$ in \textbf{d}.
}\label{fig:diagex}
\end{figure}

\begin{figure}[h!]
\includegraphics[scale=1]{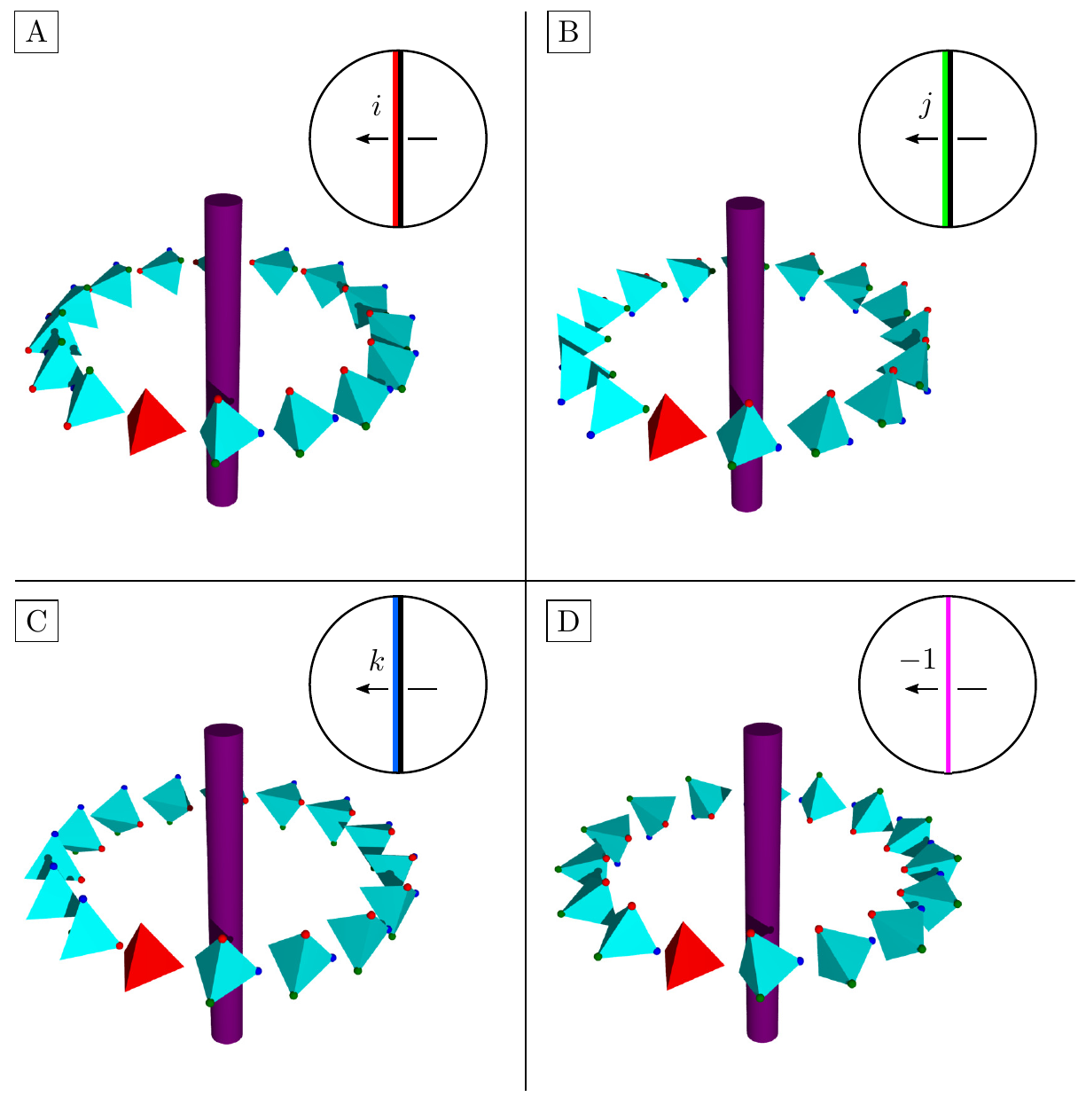}
\caption{
\textbf{Diagrams for $Q_8$-colored links, I.}
Instead of labeling the arcs of a $Q_8$-colored link diagram by arrows and elements of $Q_8$, this information is expressed graphically in terms of three types of bicolored arcs (\textbf{a-c}), and one type of colored arc (\textbf{d}). An example of the winding of the tetrahedral order parameter field about the core of a topological vortex corresponding to each type of (bi)color is depicted as well. The colored spheres at the vertices of the tetrahedra are not physical data; their purpose is to clarify the winding of the tetrahedral field. The basepoint is marked by a red tetrahedron, and the winding direction is anticlockwise. Clockwise winding corresponds to the inverse element in the group $Q_8$.
}\label{fig:rules}
\end{figure}

\begin{figure}[h!]
\includegraphics[scale=1.75]{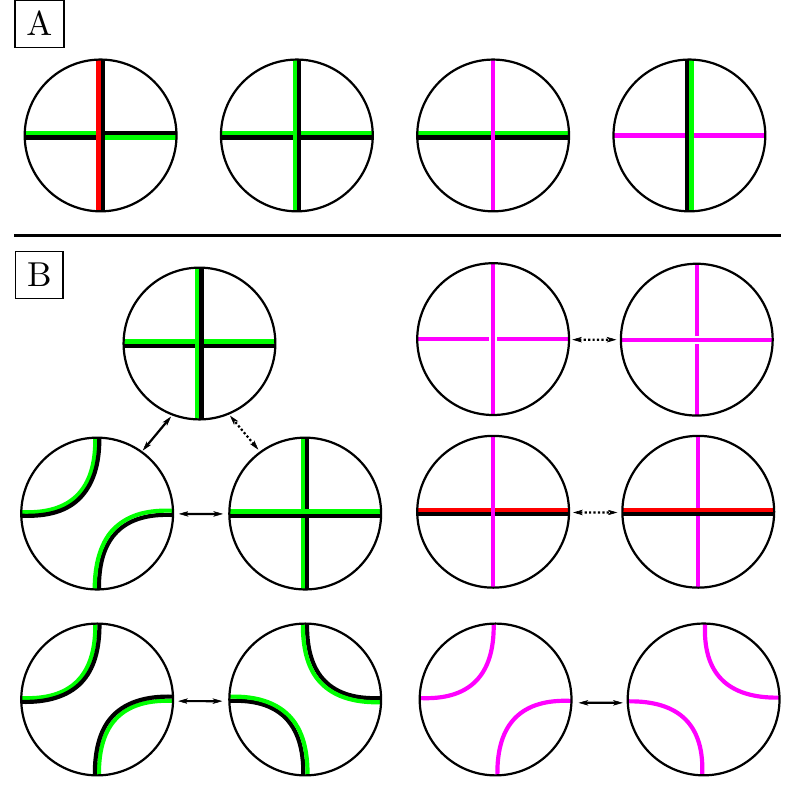}
\caption{
\textbf{Diagrams for $Q_8$-colored links, II.}
\textbf{a}, In a $Q_8$-colored diagram, the direction of the bicoloring is flipped when passing under a strand of a different color, unless it is purple. When passing under a strand of the same color, the bicoloring remains unchanged.
\textbf{b}, Core-topology-altering local modifications. The modifications marked by dashed arrows are strand crossings, which can occur only between strands of the same color, or if at least one of the strands is purple\cite{poenaru:1977}. The other modifications are local reconnections. They can also take place only between strands of the same color. Moreover, they are further restricted by the continuity of the bicoloring away from the undercrossings. 
}\label{fig:rules}
\end{figure}

\begin{figure}[h!]
\includegraphics[scale=1.25]{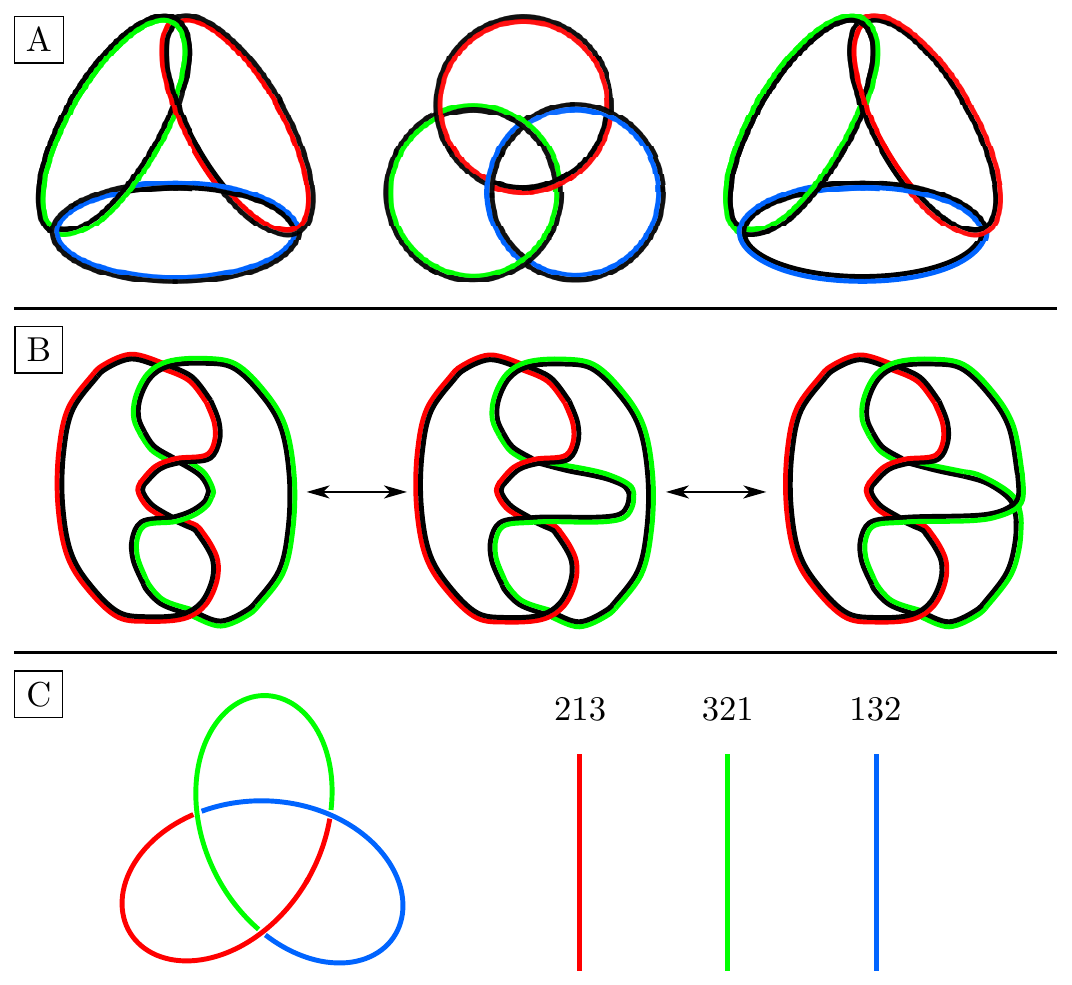}
\caption{
\textbf{Examples of topologically protected and unprotected colored knots and links.}
\textbf{a}, Three non-trivial $Q_8$-colored links and the values of the linking invariant $l$ and the $Q$-invariant for each of them. Interestingly, the topological protection the defect at the center, which is a $Q_8$-colored version of the Borromean rings, cannot be detected using the linking invariant $l$. It is established in the supplementary material that, up to topologically allowed reconnections and strand crossings, each $Q_8$-colored link with a nontrivial $Q$-invariant is equivalent to either one of the defects depicted above, or the disjoint union of one of the above defects, and a purple loop.
\textbf{b}, $Q_8$-colored link that is not topologically protected. Even though no nontrivial modifications can be applied on the crossings, the link may be unlinked by deforming one of the arcs, and then performing a local reconnection. The end result, the rightmost link diagram, is an unlink. Accordingly, the link has trivial invariants $l$ and $Q$. Note that a red loop and a blue loop could not be form the Hopf link because such a configuration does not admit a valid bicoloring. 
\textbf{c}, Topologically protected knot colored by the elements of the permutation group $S_3$. Each color corresponds to one of the three transpositions, as indicated. Transpositions are their own inverses, and therefore the direction does not matter. 
}\label{fig:defects}
\end{figure}

\appendix
 
 \
 
\textbf{Acknowledgements} We have received funding from the European Research Council under Grant No 681311 (QUESS), from the Academy of Finland Centre of Excellence program (project 336810), and from the Vilho, Yrj\"o and Kalle V\"ais\"al\"a Foundation of the Finnish Academy of Science and Letters. 

\textbf{Author contributions} The theoretical work was carried out by T.A. with input from M.M. and R.Z.Z. M.M. supervised the work. All authors discussed the theoretical results and commented on the manuscript.

\section*{Supplementary information}

\subsection{Order parameter spaces for spin-2 biaxial nematic and cyclic phases}

Here, we study the order parameter spaces of spin-2 biaxial nematic and cyclic phases, and identify the quaternion group $Q_8$ as the subgroup of the fundamental group corresponding to topological vortices with no scalar-phase winding about the core.

We begin by analyzing the order parameter space of the spin-2 biaxial nematic phase $M_\mathrm{BN}$. It is known that $M_\mathrm{BN}\cong [S^1 \times \SO(3)] / D_4$, where the $S^1$ accounts for the scalar complex phase, and $D_4$ is the symmetry group of a square lying in the $xy$-plane\cite{song:2007, turner:2007}. The dihedral group $D_4$ is realized as a subgroup of $S^1 \times \SO(3)$ in such a way that the $90^\circ$ and the $270^\circ$ rotations along the $z$-axis, and the $180^\circ$ rotations along $x+y$ and the $x-y$ axis are supplemented with a phase shift by $\pi$. In other words, the elements of the Klein four-group $K_4$, corresponding to the subgroup of $180^\circ$ rotations along the $x$-, $y$- and the $z$-axis as well as the identity, are not accompanied by phase shifts, and therefore $[S^1 \times \SO(3)] / K_4 \cong S^1 \times [\SO(3) / K_4]$. Since the inverse image of $K_4$ under the two-fold covering $\Spin(3) \to \SO(3)$ is the quaternion group $Q_8$, we deduce that $\pi_1(S^1 \times [\SO(3) / K_4]) \cong \Zb \times  Q_8$. Moreover, since $D_4/K_4 \cong \Zb_2$, the order parameter space $M_\mathrm{BN}$ is homeomorphic to the quotient $\big\{ S^1 \times [\SO(3) / K_4] \big\} / \Zb_2$. Applying the long exact homotopy sequence\cite{hatcher:2002}, we obtain a short exact sequence of groups
\[1 \to \pi_1(S^1 \times [\SO(3) / K_4]) \to \pi_1(M_\mathrm{BN}) \to \Zb_2 \to 1\]
leading us to conclude that $\Zb \times  Q_8$ is a subgroup of $\pi_1(M_\mathrm{BN})$. The elements of $\pi_1(M_\mathrm{BN})$ that do not belong to $\Zb \times Q_8$ correspond to paths between different points in a $\Zb_2$-orbit of $S^1 \times [\SO(3) / K_4]$\cite{hatcher:2002}. As the $\Zb_2$-action identifies points that have scalar phase $\alpha$ with points that have scalar phase $\alpha + \pi$, no path connecting such a pair has phase winding that is an integer multiple of $2\pi$. Therefore, the subgroup $\Zb \times Q_8 \subset \pi_1(M_\mathrm{BN})$ corresponds to those vortices that have integer phase winding, and $\{0\} \times  Q_8 \subset \pi_1(M_\mathrm{BN})$ corresponds to exactly those vortices that have no scalar phase winding.

Next, we analyze the order parameter space of the spin-2 cyclic phase in detail. The order parameter space $M_\mathrm{C}$ is $[S^1 \times \SO(3)] / T$, where the $S^1$, again, accounts for the complex phase, and where $T$ is the group of rotational symmetries of a tetrahedron\cite{semenoff:2007}. To realize it as a subgroup of $S^1 \times \SO(3)$, we employ the presentation of $\SO(3)$ as the group of $3 \times 3$ orthogonal matrices with determinant $1$. Then, the elements corresponding to the Klein four-group $K_4 \subset T$, namely $I, \diag(-1,-1,1), \diag(-1,1,-1)$ and $\diag(1,-1,-1)$ are not accompanied by any phase shifts, and the tetrahedral group is generated by the above elements in $\{0\} \times \SO(3) \subset S^1 \times \SO(3)$, and by
\[
\Bigg(\frac{2 \pi}{3}, 
\begin{bmatrix}
0 & 1 & 0 \\
0 & 0 & 1 \\
1 & 0 & 0
\end{bmatrix}
\Bigg) \in S^1 \times \SO(3).
\]
As in the previous paragraph, $[S^1 \times \SO(3)] / K_4 \cong S^1 \times [\SO(3) / K_4]$ and its fundamental group is isomorphic to $\Zb \times Q_8$. The Klein four-group is a normal subgroup of $T$ of index 3 and the quotient $T / K_4$ is the cyclic group $\Zb_3$. Applying the long exact homotopy sequence to the three-fold covering $S^1 \times [\SO(3) / K_4] \to M_\mathrm{C}$\cite{hatcher:2002}, we obtain a short exact sequence of groups
\[1 \to \pi_1(S^1 \times [SO(3) / K_4]) \to \pi_1(M_\mathrm{C}) \to \Zb_3 \to 1.\]
 The elements of $\pi_1(M_\mathrm{C})$ that do not belong to $\Zb \times Q_8$ correspond to paths between different points in a $\Zb_3$-orbit of $S^1 \times [\SO(3) / K_4]$\cite{hatcher:2002}. As the $\Zb_3$-action identifies points that have scalar phase $\alpha$ with points that have scalar phase $\alpha + 2\pi/3$ and $\alpha + 4 \pi/3$, no path connecting such a pair has phase winding that is an integer multiple of $2\pi$. Therefore, the subgroup $\Zb \times Q_8 \subset \pi_1(M_\mathrm{C})$ corresponds to those vortices that have integer phase winding, and $\{0\} \times  Q_8 \subset \pi_1(M_\mathrm{C})$ corresponds to exactly those vortices that have no scalar phase winding.

\subsection{The $Q$-invariant and the classification of $Q_8$-colored links}

Here, we define the $Q$-invariant of $Q_8$-colored links, and use it to study their classification. The invariant is $\Zb_4$-valued, and recovers the linking invariant $l$ when reduced to modulo 2. The definition of $Q$ requires focusing on loops of either red, green, or blue color, and therefore one ends up with three colored invariants $Q_\red, Q_\green$, and $Q_\blue$, the equivalence of which is established later. After defining the colored invariants of a $Q_8$-colored link diagram, we establish their independence from the specific diagram chosen to present a $Q_8$-colored link, after which we prove that these invariants are conserved in topologically allowed strand crossings and reconnections. Subsequently, we employ the classification of three-component links to classify $Q_8$-colored links up to strand crossings and reconnections and establish the equivalence of the three colored invariants. The $Q$-invariant is thus defined as the value of any of the colored invariants. Moreover, we establish a relationship of $Q$ with Milnor's triple linking number.

Suppose $L$ is a $Q_8$-colored link represented by a $Q_8$-colored link diagram. Let $L_1,...,L_r$ be the components of $L$ enumerated in some order. For each non-purple $L_i$ choose a \emph{basepoint} $b_i$ at one of the arcs belonging to the loop $L_i$ in the diagram. We will use the following orientation convention for bicolored loops with a basepoint: the loop $L_i$ is oriented in such a way that when moving from the basepoint $b_i$ according to the orientation, the black color of the bicoloring is on the right. It will not be necessary to choose basepoints or specify orientations for the purple loops.


\begin{defn}\label{def:ptdalpha}
Let $(L_i,b_i)$ be a pointed loop of color $c \in \{\red, \green, \blue \}$ in a $Q_8$-colored link diagram. We define the $\alpha$-invariant of $(L_i,b_i)$ as
\begin{equation}\label{eq:ptdalpha}
\alpha(L,b) := \chi_c^{-\omega} q,
\end{equation}
where $\chi_c$ is $i$, $j$ or $k$ if $c$ is $\red$, $\green$ or $\blue$, respectively; $\omega$ is the \emph{self-writhe} of the loop $L$, i.e., the signed count of self crossings of $L$ in the diagram, where the sign of a crossing is decided using the right-hand rule, see Extended Data Figure 8; and $q \in Q_8$ is obtained by multiplying the quaternions, corresponding to the crossings of $L_i$ under non-purple strands, in order from right to left, when $L_i$ is traversed from the basepoint $b_i$ according to the orientation specified above. Concrete examples are presented in Extended Data Figure 1.
\end{defn}

Let us record the following useful properties of the $\alpha$-invariant.

\begin{lem}\label{lem:alpha}
The $\alpha$-invariant satisfies the following properties:
\begin{enumerate}
\item $\alpha(L_i,b_i)$ is a power of $\chi_c$;
\item $\alpha(L_i,b_i)$ does not depend on the basepoint $b_i$.
\end{enumerate}
\end{lem}
\begin{proof}
Throughout the proof, we ignore all the purple strands since they do not affect the invariant. In order for the bicoloring to be consistent, each loop is overcrossed by a strand of different color an even number of times. For instance, if the loop $L_i$ is red, then, in the expression of $\alpha(L_i,b_i)$, the combined number of occurrences of $j$ and $k$ is even. After reordering, which contributes only a sign, the occurrences of $j$ and $k$ can be replaced by a single power of $i$, which proves the first claim.

For the proof of the second claim, we investigate the effect on $\alpha$ caused by moving the basepoint through an undercrossing to an adjacent arc in the diagram. In fact, as the self-writhe does not depend on the orientation of $L_i$, we need only to consider how this process affects $q$. There are two cases to consider, depending on the color $c'$ of the strand crossing over $L_i$. 
\begin{enumerate}
\item \emph{Case $c' = c$}: the basepoint moves, but the orientation is not altered. If the undercrossing through which the basepoint is moved is the first undercrossing according to the orientation, then $q = q' \chi^{\pm 1}_c$ is replaced by $\chi_c^{\pm 1} q'$. However, since $q' = q \chi_c^{\mp 1}$ is a power of $\chi_c$, it commutes with $\chi_c$, so $q$ does not change. The other case of the moving the base point along the negative orientation of the loop is proved in a similar fashion.
 
\item \emph{Case $c' \not = c$}: the basepoint moves and the orientation is altered. If the undercrossing is the first undercrossing according to the orientation, then $q = q' \chi_{c'}^{\pm 1}$ is replaced by $(\chi_{c'}^{\pm 1} q')^{-1} = q'^{-1} \chi_{c'}^{\mp 1}$. However, since $q'\chi_{c'}^{\pm}$ is a power of $\chi_c$, $q' \not \in \{1, -1\}$, and therefore $q'^{-1} = -q'$. Hence $q'^{-1} \chi_{c'}^{\mp 1} = q' \chi_{c'}^{\pm 1}$, and hence $q$ does not change in the process. The other case of the moving the base point along the negative orientation of the loop is proved in a similar fashion. \qedhere
\end{enumerate}
\end{proof}

Below, we define the colored invariants.

\begin{defn}\label{def:colinv}
Let $c \in \{ \red, \green, \blue \}$. Then the \emph{colored invariant} $Q_c \in \Zb_4$ of a $Q_8$-colored link diagram is defined as
\begin{equation}\label{eq:colinv}
\chi_c^{Q_c} := (-1)^{l_c} \prod_{L_i \text{ of color $c$}} \alpha(L_i),
\end{equation}
where $l_c$ is the sum of the pairwise linking numbers between loops of color $c$, modulo $2$. The product on the right side is well defined, since $\alpha(L_i)$ does not depend on the choice of a basepoint of $L_i$ and since the $\alpha(L_i)$ commute with each other as they are all powers of $\chi_c$. The invariant is well defined since the exponent of $\chi_c$ is well defined modulo $4$. Concrete examples are illustrated in Extended Data Figure 1.
\end{defn}

Our next result establishes a relationship between the invariants $Q_c$ and the linking invariant $l$.

\begin{prop}\label{prop:colinvgauss}
When reduced to modulo 2, the invariant $Q_c$ is equivalent to the linking invariant $l$.
\end{prop}
\begin{proof}
Recall that the linking invariant is the number of times a strand of color $c'$ passes over a strand of color $c''$, modulo 2, where $c'$ and $c''$ are two different non-purple colors. Let $c$ be color that is different from $c'$ and $c''$. As there are an even number of crossings between different loops of color $c$, these contribute an even number of multiplicative factors of $\pm \chi_c$ to $\prod_{L_i} \alpha(L_i)$, where $L_i$ ranges over all loops of color $c$. Moreover, the same is true for each self crossing of a loop of color $c$, as each self crossing contributes both to terms $\chi_c^{-\omega}$ and $q$ in Equation (\ref{eq:ptdalpha}). In other words, $\chi_c^{Q_c}$ is a product of quaternions, an even number of which are $\pm \chi_c$. 

If $l = 1$, then, in the expression of $\chi_c^{Q_c}$, there are an odd number of factors of form $\pm\chi_{c'}$ and an odd number of factors of form $\pm \chi_{c''}$. Reordering the expression, we conclude that $\chi_c^{Q_c} = \pm \chi_c$, i.e., $Q_c \equiv 1$ modulo 2. Similarly, if $l = 0$, then, in the expression of $\chi_c^{Q_c}$, there are an even number of factors of form $\pm\chi_{c'}$ and an even number of factors of form $\pm \chi_{c''}$. Reordering the expression, we conclude that $\chi_c^{Q_c} = \pm 1$, i.e., $Q_c \equiv 0$ modulo 2.
\end{proof}
 
There are various diagrams representing the same abstract link in $\Rb^3$, and they are related to each other by Reidemeister moves\cite{rolfsen:2003}. Similarly, there are many $Q_8$-colored link diagrams representing the same abstract $Q_8$-colored link, and they are related to each other by Reidemeister moves. Given an initial coloring, there exists a unique $Q_8$-coloring for the link diagram after a Reidemeister move has taken place; we refer to such moves, endowed with the data of a $Q_8$-coloring, as $Q_8$-colored Reidemeister moves. Examples are presented in Extended Data Figure 2. Next we prove that the invariants $Q_c$ are invariants of the $Q_8$-colored link rather than the particular link diagram chosen to present it.

\begin{lem}\label{lem:colinvreid}
The invariants $Q_c$ are conserved in $Q_8$-colored Reidemeister moves.
\end{lem}
\begin{proof}
We consider each type of Reidemeister move separately. The essential cases one has to consider are presented in Extended Data Figure 2. We again ignore the purple strands of the diagram in the proof, for they have no effect on the invariants.
\begin{enumerate}
\item \emph{Reidemeister move of type I}: such a move has the potential to alter the invariant only if it is applied to a loop $L_i$ of color $c$. Let $L'_i$ be the loop after the move has been performed. Choosing a suitable basepoint, one obtains expressions $\alpha(L_i) = \chi_c^{-\omega} q$ and $\alpha(L'_i) = \chi_c^{-\omega'} q' = \chi_c^{-\omega \mp 1} \chi_c^{\pm 1} q$, proving that $\alpha(L'_i) = \alpha(L_i)$. As the operation does not alter the linking numbers between loops of color $c$, the invariant $Q_c$ remains unchanged.

\item \emph{Reidemeister move of type II}: such a move does not alter the self-writhe, the $\alpha$-invariant of any loop or the linking numbers between loops, so the invariant remains unchanged.

\item \emph{Reidemeister move of type III}: such a move does not alter the self-writhe, the $\alpha$-invariant of any loop or the linking numbers between loops, so the invariant remains unchanged. \qedhere
\end{enumerate}
\end{proof}

Next, we investigate the conservation of the colored invariants $Q_c$ under topologically allowed strand crossings and local reconnections. As a disjoint union of $Q_8$-colored unknotted loops has trivial invariants $Q_c$, these invariants may be regarded as an obstructions for the unlinking of a $Q_8$-colored link using local reconnections and strand crossings.

\begin{lem}\label{lem:crossinv}
The invariants $Q_c$ are conserved in topologically allowed strand crossings. 
\end{lem}
\begin{proof}
The proof is presented in Extended Data Figure 3.
\end{proof}

\begin{lem}\label{lem:reconninv}
The invariants $Q_c$ are conserved in topologically allowed reconnections.
\end{lem}
\begin{proof}
If the reconnection takes place between two different loops, they will merge into one loop as a result. As it is enough to establish the conservation of $Q_c$ in the inverse process of such an event, we may assume that the reconnection takes place between points $x$ and $y$ on the same loop $L_i$. Moreover, since any knot can be unknotted by crossing changes, Lemma \ref{lem:crossinv} implies that we may assume that $L_i$ is an unknot. Hence, we may choose a $Q_8$-colored link diagram representing the link $L$ where the loop $L_i$ has no self-crossings, and where no extra arcs appear in the imminent neighborhood of the location where the reconnection takes place. Visualization is provided by Extended Data Figure 4. 

Let $q_1$ and $q_2$ be the quaternions defined much like $q$ in Definition \ref{def:ptdalpha}, but where $q_1$ accounts for the undercrossings occurring when traversing from $x$ to $y$ according to the orientation convention at $x$, and where $q_2$ accounts for the undercrossings occurring when continuing from $y$ back to $x$. Note that $\alpha(L_i) = q_2 q_1$. There are two cases to consider.
\begin{enumerate}
\item \emph{$q_1$ is a power of $\chi_c$}: i.e., on the path from $x$ to $y$, $L_i$ crosses under an even number of strands than of color other than $c$. In this case, the reconnection splits the loop $L_i$ into two loops $L'_i$ and $L''_i$. Moreover, as $\alpha(L'_i) = q_2$ and $\alpha(L''_i) = q_1$, and as the process does not affect $l_c$ modulo 2, the invariant $Q_c$ remains unchanged.
\item \emph{$q_1$ is not a power of $\chi_c$}: by rotating the diagram, by rotating the loop $L_i$ around an axis in the plane of the diagram, and by swapping the labels of $x$ and $y$ if necessary, we may assume that the strands are horizontal around the point of reconnection, that $x$ is below $y$ in the picture, and that the orientation at $x$ points to right. There are two possible reconnections which are related by a topologically allowed strand crossing. By Lemma \ref{lem:crossinv}, it is enough to investigate only one of these; we focus on the one in which, after traversing the path corresponding to $q_1$ in the modified loop $L'_i$, the loop crosses under itself. As the self writhe of $L'_i$ is $-1$, by definition $\alpha(L'_i) = \chi_c q_2^{-1} \chi_c^{-1} q_1$. As $q_1$ is not a power of $\chi_c$ but $q_2 q_1$ is, $q_2^{-1} = -q_2$ and $q_2$ anticommutes with $\chi_c$, and therefore $\chi_c q_2^{-1} \chi_c^{-1} q_1 = q_2 q_1$. In other words, $\alpha(L'_i) = \alpha(L_i)$. As the reconnection does not affect the total modulo-two linking number between loops of color $c$, the invariant $Q_c$ remains unchanged. \qedhere
\end{enumerate}
\end{proof}

Using the previous results, it is possible to classify all the $Q_8$-colored links up to strand crossings and local reconnection events (we refer to this, in short, as the \emph{classification of $Q_8$-colored links}), and to prove the equivalence of the colored invariants $Q_\red$, $Q_\green$ and $Q_\blue$. Given a $Q_8$-colored link, one can perform local reconnections in order to connect loops of the same color. After doing so and ignoring the potential purple loop, we have a link of at most three components, and each component is labeled with either red, green or blue color. We will momentarily forget the bicoloring, and choose an orientation for each loop.

Milnor has classified oriented links of at most three components up to \emph{link homotopy}\cite{milnor:1954, milnor:1957}, i.e., up to such continuous deformations of the link where the different components are not allowed to meet, but where self intersections are allowed. Such deformations may be expressed in terms of $Q_8$-colored Reidemeister moves and topologically allowed strand crossings, and therefore Milnor's classification provides an intermediate step in the classification of $Q_8$-colored links. If there is only one loop, then there is only one link up to link homotopy. If there are two loops, $L_1$ and $L_2$, then the link is completely characterized by the Gauss linking number $\mu(12)$ of $L_1$ and $L_2$. In case of three loops, the third loop $L_3$ corresponds to a canonical element of form $\alpha_1 \alpha_2^{-\mu(123)} \alpha_1^{-1} \alpha_2^{\mu(123)}\alpha_2^{\mu(23)} \alpha_1^{\mu(13)} $ in the fundamental group $\pi_1(\Rb^3 \backslash (L_1 \cup L_2))$, where $\alpha_i$ is an element of $\pi_1(\Rb^3 \backslash (L_1 \cup L_2))$ that corresponds to a loop that winds once around $L_i$ in the positive direction. A concrete example is provided in Extended Data Figure 6. Above, $\mu(ij)$ is the linking number between $L_i$ and $L_j$, and $\mu(123)$ is the \emph{triple linking number}, an integer well defined up to the greatest common divisor $d$ of $\mu(12), \mu(13)$ and $\mu(23)$. The link is completely classified by $\mu(12), \mu(13), \mu(23) \in \Zb$ and $[\mu(123)] \in \Zb_d$.

There are three non-trivial cases to consider.
\begin{enumerate}
\item \emph{The link has only two loops.} The bicoloring cannot be consistent unless $\mu(12)$ is even. Applying the surgery operation depicted in Extended Data Figure 5 (A), it is possible to achieve $\mu(12) = 0$. In other words, such a $Q_8$-colored link is in the trivial class. 

\item \emph{The link has three loops and $\mu(12)$ is odd.} In order for the bicoloring to be consistent, also $\mu(13)$ and $\mu(23)$ have to be odd. Applying the surgery operation depicted in Extended Data Figure 5 (A), it is possible to achieve $\mu(12)=\mu(13)=\mu(23) = 1$. These numbers classify the link up to link homotopy as their greatest common divisor is 1; such a link is homotopic to a looped chain of length three. There are 8 possible bicolorings for such a link. It is possible to ``flip'' any two bicolorings simultaneously, as depicted in Extended Data Figure 5 (C), leaving exactly two classes, which correspond to values $Q_c = [1]$ and $Q_c = [3]$.

\item \emph{The link has three loops and $\mu(12)$ is even.} In order for the bicoloring to be consistent, also $\mu(13)$ and $\mu(23)$ have to be even. Applying the surgery operation depicted in Extended Data Figure 5 (A), it is possible to achieve $\mu(12)=\mu(13)=\mu(23) = 0$ and that $\mu(123)$ is either $0$ or $1$, depending on the parity of $\mu(123)$ in the original link. If $\mu(123)=0$, then the loops are not linked, so this case corresponds to the trivial class. If $\mu(123)=1$, then the link is homotopic to the Borromean rings. There are again 8 possible bicolorings for such a link. Moreover, it is possible to ``flip'' any single bicoloring, as depicted in Extended Data Figure 5 (D), so all of the possible bicolorings are in the same class which corresponds to $Q_c= [2]$.
\end{enumerate}

There are several immediate consequences for the discussion above. First we consider the equivalence of the colored invariants.

\begin{prop}
The invariants $Q_\red, Q_\green$ and $Q_\blue$ are equivalent. 
\end{prop}
\begin{proof}
According to the above discussion, any nontrivial $Q_8$-colored link diagram can be reduced to one of the three links depicted in Extended Data Figure 5 (B) by applying topologically allowed strand crossings and reconnections. Hence, one has to check the desired equality $Q_\red = Q_\green = Q_\blue$ only in these cases.
\end{proof}

The invariant $Q$ is defined as the common value of the colored invariants $Q_c$. The following result establishes a connection between the it and Milnor's triple linking number.

\begin{prop}
If $L$ is a $Q_8$-colored link with at most one component of each color. If $l = 0 \in \Zb_2$, then $Q = [2 \mu(123)] \in \Zb_4$, where $\mu(123)$ is Milnor's triple linking number.
\end{prop}
\begin{proof}
This was established in the third case of the above discussion.
\end{proof}

Finally, we classify $Q_8$-colored links. 

\begin{thm}[Classification of $Q_8$-colored links]\label{thm:classification} Up to topologically allowed strand crossings and reconnections, there are only the following classes of $Q_8$-colored links:
\begin{enumerate}
\item 16 classes of trivial links, corresponding to untangled disjoint unions of loops of different colors;
\item 2 classes each (with and without a purple loop) for which the invariant $Q$ obtains the values $[1]$, $[2]$ and $[3]$ in $\Zb_4$.
\end{enumerate}
\end{thm}
\begin{proof}
This follows immediately from the discussion above.
\end{proof}

\begin{figure}
\includegraphics[scale=1.1]{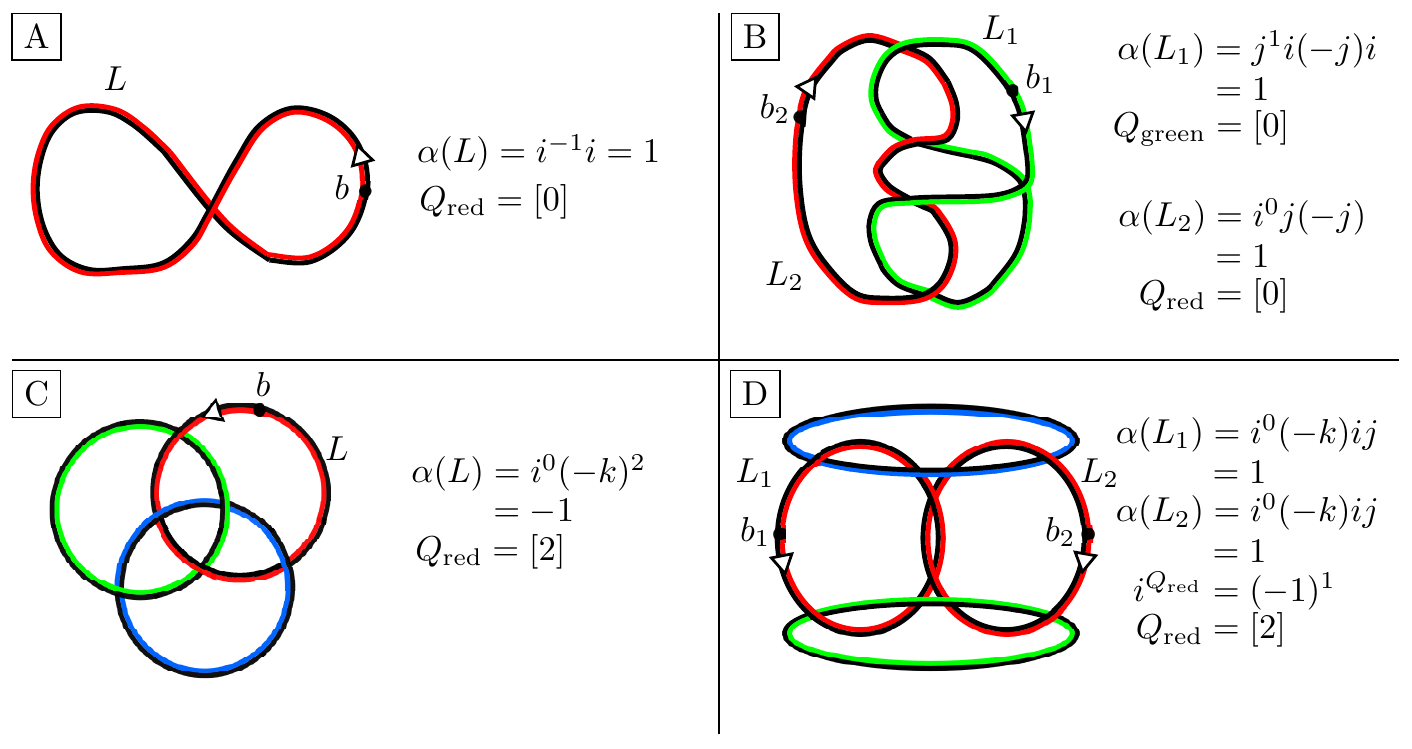}
\begin{flushleft}
\textbf{Extended Data Figure 1} | \textbf{Examples of $\alpha$-invariants and colored invariants $Q_c$ of $Q_8$-colored links.} \textbf{a}--\textbf{d}, Loops of interest $L_k$ and the corresponding basepoints $b_k$ together with their $\alpha$ and $Q_c$ invariants. Recall that according to the orientation convention, a loop is oriented in such a way, that when moving from the basepoint according to the orientation, the black color of the bicoloring is on the right.
\end{flushleft}
\end{figure}

\begin{figure}
\includegraphics[scale=1.1]{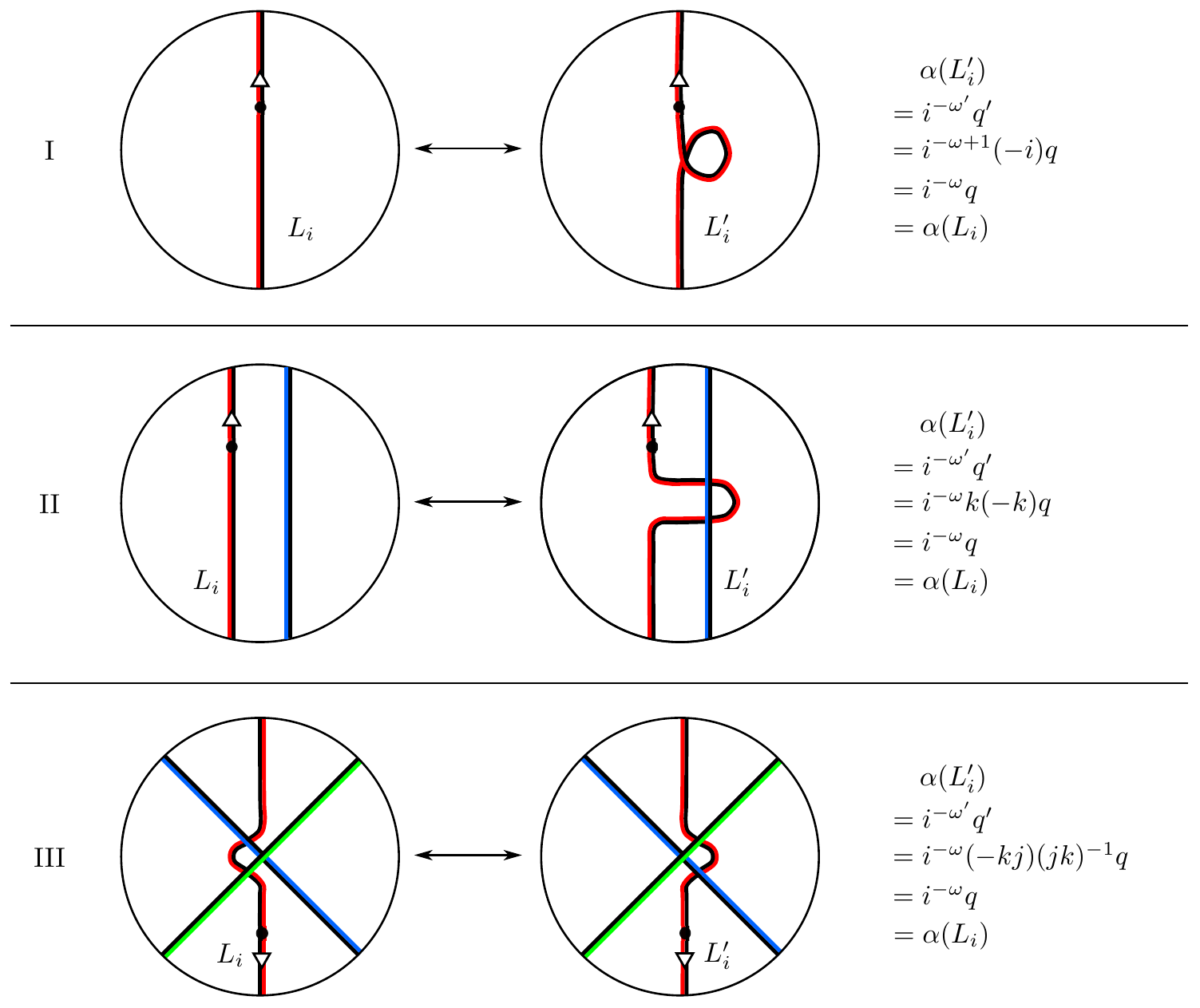}
\begin{flushleft}
\textbf{Extended Data Figure 2} | \textbf{Examples of $Q_8$-colored Reidemeister moves, and the conservation of the colored invariants $Q_c$.} The Reidemeister moves do not affect the pairwise modulo-two linking numbers between loops of the same color, and therefore it is enough to verify that the $\alpha$-invariant of the modified loop remains unchanged. 
\end{flushleft}
\end{figure}

\begin{figure}
\includegraphics[scale=0.9]{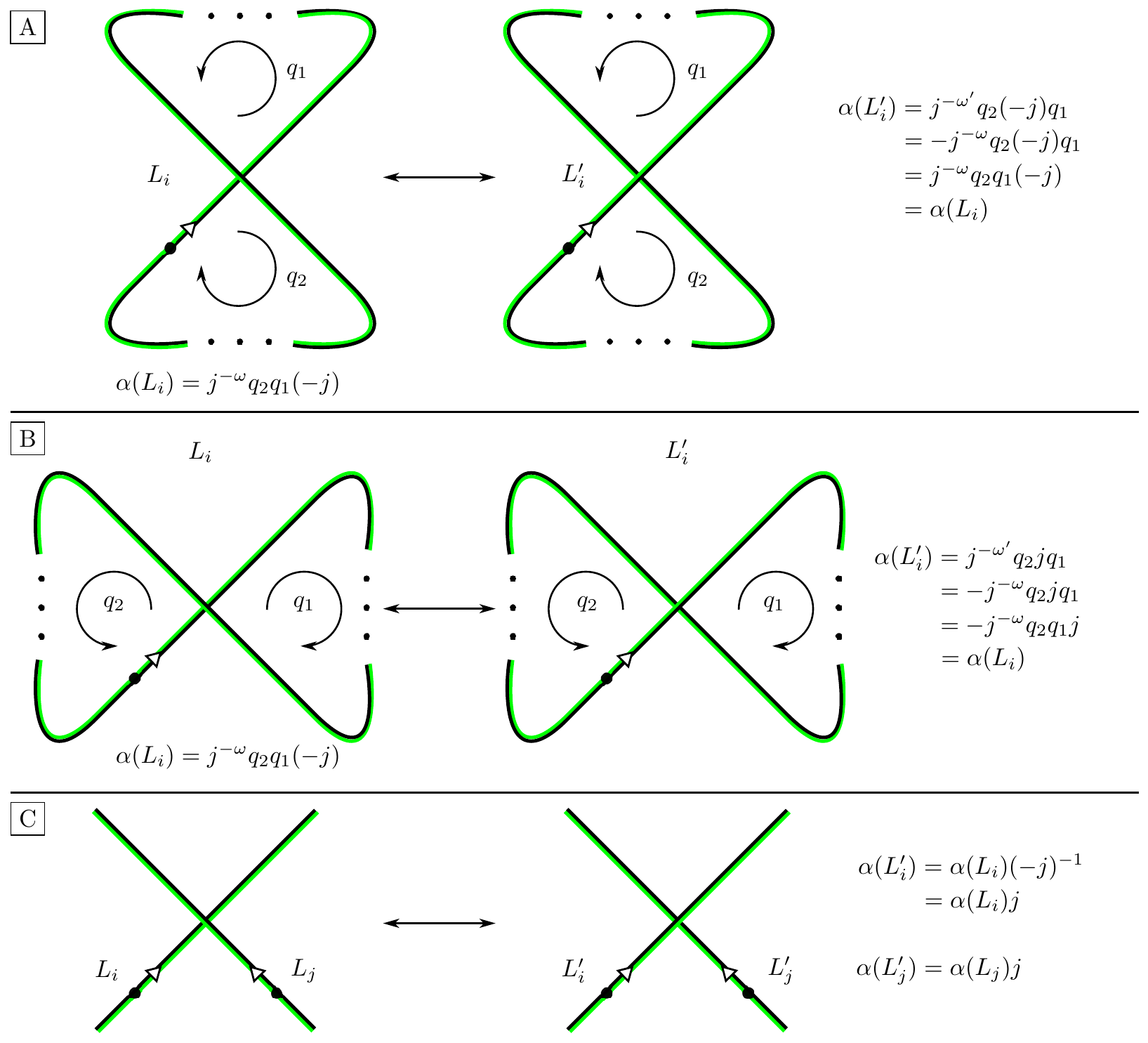}
\begin{flushleft}
\textbf{Extended Data Figure 3} | \textbf{Conservation of the colored invariant $Q_\green$ in topologically allowed strand crossings.} By rotating the picture if necessary, one can assume that the black color of the bicoloring is on the right hand side when moving upwards along either of the strands of the crossing. The pictures (A), (B) and (C) depict the cases that can occur. In (A) and (B) the strands partaking in the crossing change are parts of the same loop $L_i$; note that the self-writhes of $L_i$ and $L'_i$ differ by $\pm 2$. The difference between the two cases is that in (A), the quaternion $q_1$ anticommutes with $j$, because the path corresponding to $q_1$ passes under an odd number of red and green strands in total, whereas, for analogous reasons, in (B), $q_1$ commutes with $j$. The crossing change in (A) and (B) does not alter the linking number of $L_i$ with other red loops, so $Q_\green$ remains unchanged. In (C), the strands are part of loops $L_i$ and $L_j$. Note that $\alpha(L'_i)\alpha(L'_j) = - \alpha(L_i)\alpha(L_j)$, but the introduced sign is canceled as the total modulo-two linking number between green loops is altered by $1$, so $Q_\green$ remains unchanged.
\end{flushleft}
\end{figure}

\begin{figure}
\includegraphics[scale=1.1]{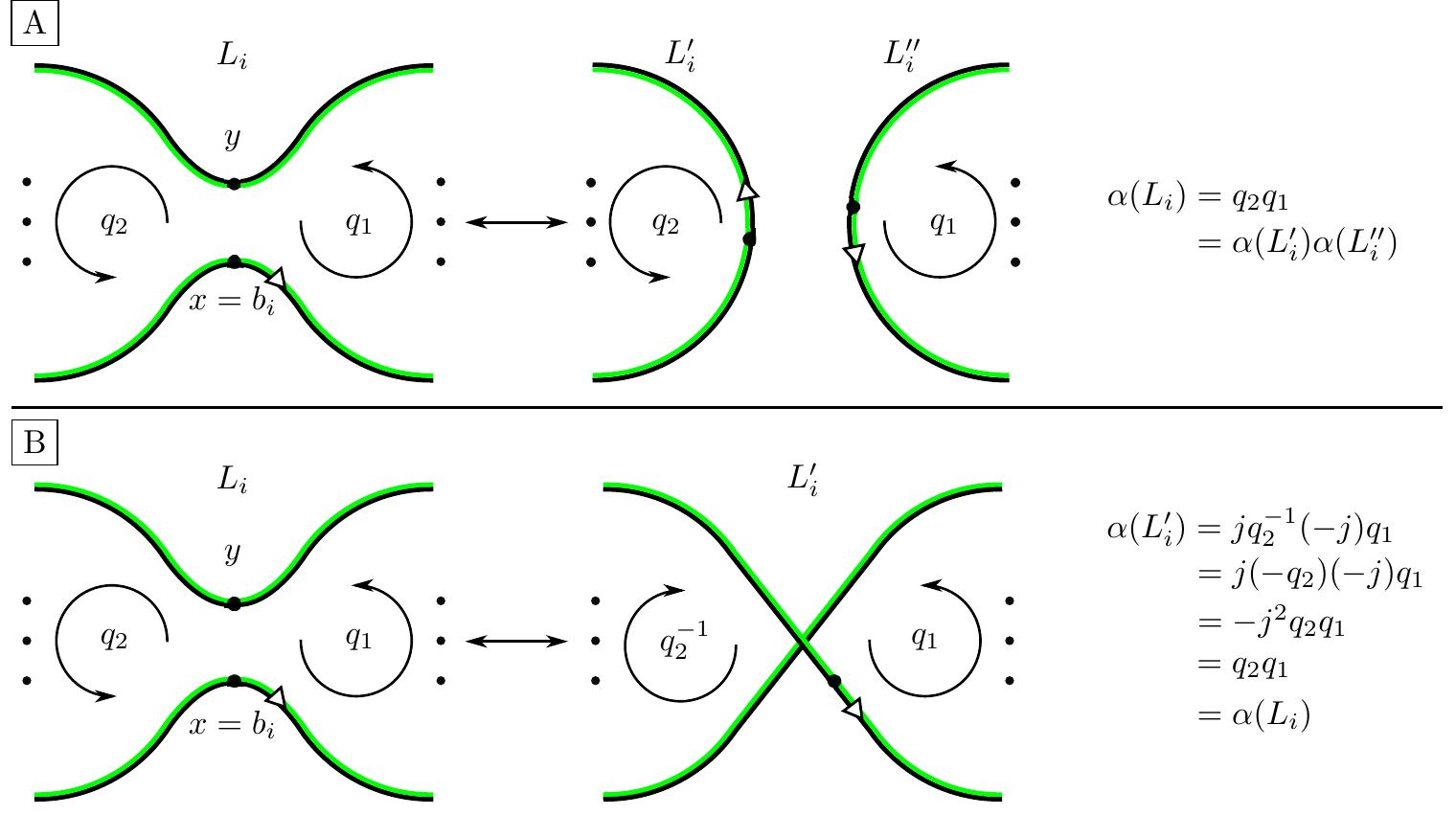}
\begin{flushleft}
\textbf{Extended Data Figure 4} | \textbf{Conservation of the colored invariant $Q_\green$ in topologically allowed reconnections.} (A) and (B) illustrate the cases to consider. The difference between them is the total number of times $L_i$ crosses under a red or a blue strand, modulo 2, when traversing from the point $x$ to $y$. (A) In the case of an even number, the loop $L_i$ splits into two loops $L'_i$ and $L''_i$. (B) In the case of an odd number, the quaternions $q_i$ anticommute with $j$ and $q_i^{-1} = -q_1$, which are important facts used to deduce the conservation of the invariant $Q_\green$.
\end{flushleft}
\end{figure}

\begin{figure}
\includegraphics[scale=1]{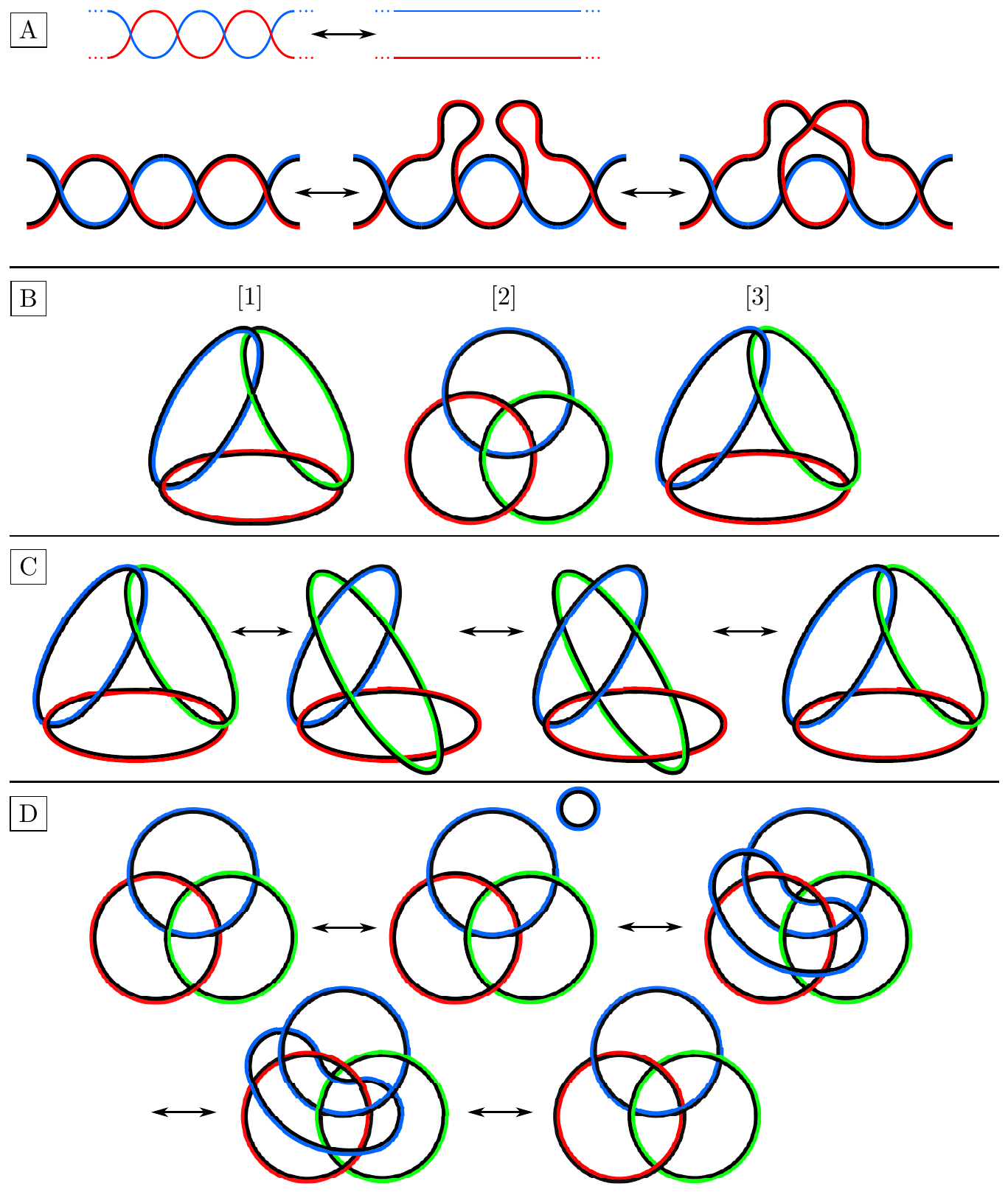}
\begin{flushleft}
\textbf{Extended Data Figure 5} | \textbf{The classification of $Q_8$-colored links.} 
(A) A surgery operation that can be employed to alter the linking number between two loops by an integer multiple of two. 
(B) The nontrivial classes of $Q_8$-colored links, and the values of the invariant $Q$ they correspond to. 
(C) In a looped chain configuration, it is possible to flip two of the three bicolorings by letting one of the loops go around the other two.
(D) In a Borromean rings configuration, it is possible to flip any single bicoloring. In the picture, the blue loop is split into two loops, one of which goes around the red loop, resulting in the flipping of the red bicoloring. The two blue loops cross each other during this process.
\end{flushleft}
\end{figure}

\begin{figure}
\includegraphics[scale=0.8]{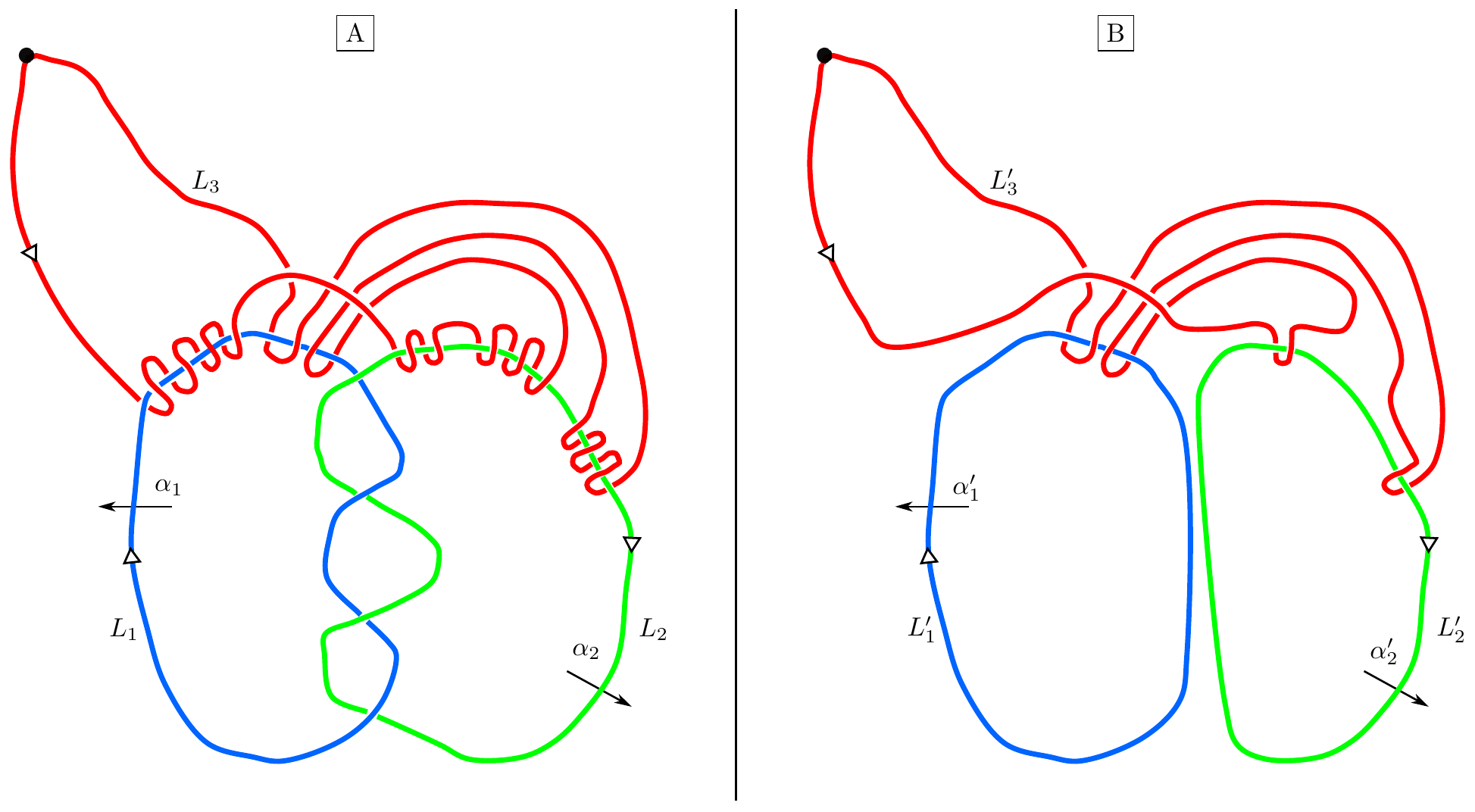}
\begin{flushleft}
\textbf{Extended Data Figure 6} | \textbf{Classification of three-component links.} 
The bicoloring is ignored for simplicity. (A) The linking number between $L_1$ and $L_2$ is $-2$, and the loop $L_3$ corresponds to an element of form $\alpha_1 \alpha_2^3 \alpha_1^{-1} \alpha_2^{-3} \alpha_2^{-2} \alpha_1^{-4}$ in $\pi_1(\Rb^3 \backslash (L_1 \cup L_2))$. (B) The same link after applying the surgery operation depicted in the Extended Data Figure 5 (A) several times. For the resulting link, the pairwise linking numbers are 0, and the triple linking number $\mu(123)$ is 1, so it is homotopic to the Borromean rings.
\end{flushleft}
\end{figure}

\begin{figure}
\includegraphics[scale=1.1]{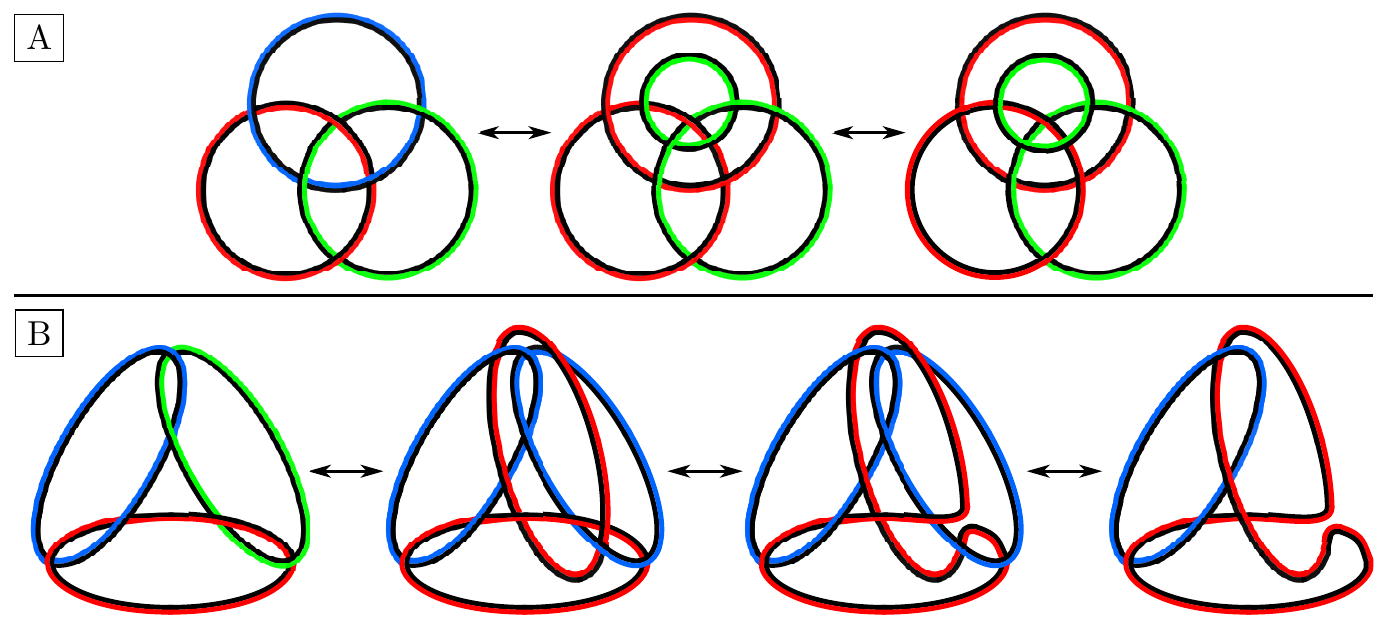}
\begin{flushleft}
\textbf{Extended Data Figure 7} | \textbf{Decay of nontrivial $Q_8$-colored links via vortex splitting.} 
Any $Q_8$-colored link may be unlinked using topologically allowed reconnections, strand crossings and vortex splittings. 
(A) Decay of the $Q_8$-colored Borromean ring configuration, corresponding to $Q = [2]$.
(B) Decay of the $Q_8$-colored looped chain of length three, corresponding to $Q = [3]$. Defects with $Q = [1]$ can be unlinked in a similar fashion.
\end{flushleft}
\end{figure}

\begin{figure}
\includegraphics[scale=1.2]{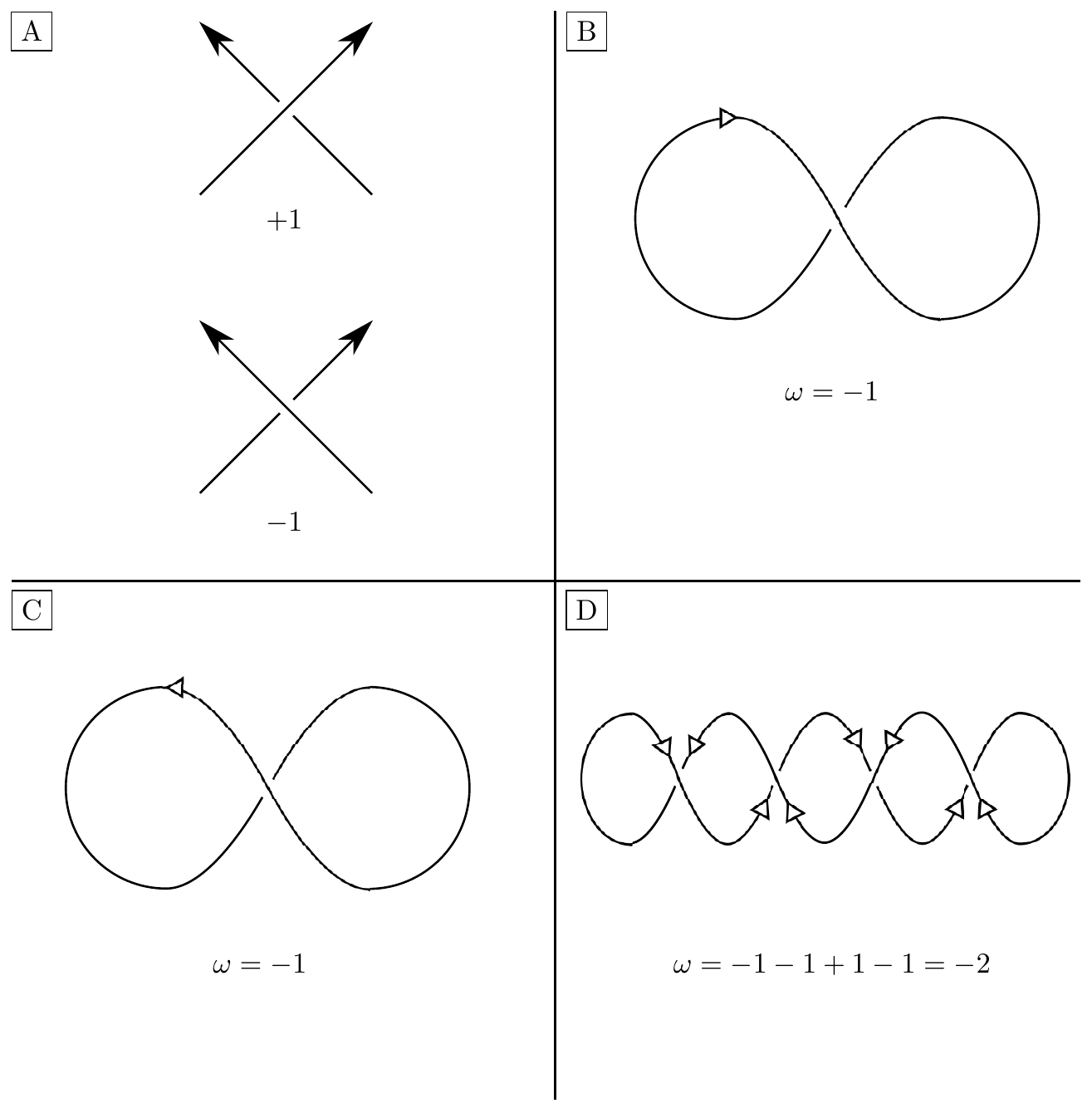}
\begin{flushleft}
\textbf{Extended Data Figure 8} | \textbf{Sign of a crossing and self writhe.} 
\textbf{a}, Right-hand rule for the sign of the oriented crossing.
\textbf{b}, \textbf{c}, The self-writhe of an oriented loop is the sum of the signs of the self-crossings in the diagram. In order to compute it, an orientation must be chosen for the loop; however, the end result does not depend on this choice.
\textbf{d}, Self writhe of a loop with four self-crossings.
\end{flushleft}
\end{figure}

\end{document}